%% file: main.tex
\documentclass[a4paper,onecolumn,11pt,unpublished]{quantumarticle}
\pdfoutput=1
\input{style}

\title{A New Method For Manipulating Circuits, Application To Quantum Adders}
\author{William Schober}
\author{Scott Wesley}
\date{August 2026}

\begin{document}

\maketitle

\begin{abstract}
    We use a new technique for manipulating controlled quantum circuits to convert between two distinct types of quantum adders, one based on the Quantum Fourier Transform and the other based on the Ripple-Carry technique from classical reversible logic.
    This conversion takes the form of an explicit gate-level transpilation.
    We also present a new quantum adder with a natural interpretation as a kind of Carry-Lookahead adder that uses no ancillas.
\end{abstract}

\input{1_introduction}

\input{2_qft_sandwich}
\input{3_qft_to_cla}
\input{4_cla_to_rc}
\input{5_conclusion}

\bibliographystyle{quantum}
\bibliography{schober_bibliography,custom_bibliography}

\newpage
\appendix
\input{appendix}

\end{document}

%% file: style.tex
\usepackage{graphicx} 
\usepackage{amsmath}
\usepackage{amsfonts}
\usepackage{amssymb}
\usepackage{amsthm}
\usepackage{framed} 
\usepackage{thmtools} 
\usepackage{subcaption}

\usepackage{algorithm} 
\usepackage[noend]{algpseudocode} 

\usepackage{xfrac} 
\usepackage{anyfontsize} 
\usepackage{braket} 
\usepackage{xspace} 

\usepackage{xcolor} 
\usepackage{lipsum} 
\usepackage{layout} 

\usepackage{hyperref}
\usepackage[capitalise]{cleveref}

\usepackage{tikz}
\usetikzlibrary{quantikz2}

\usetikzlibrary{external}
\newcommand{\circuitimg}[1]{\raisebox{-0.5\height}{\includegraphics{#1}}} 

\tikzcdset{
  classical gap/.initial=0.05cm,
}

\tikzset{
ggroup/.append style={inner sep=0pt},
phase label/.append style={label position=90},
}

\theoremstyle{plain}
\newtheorem{theorem}{Theorem}[section]

\theoremstyle{definition}
\newtheorem{proposition}[theorem]{Proposition}
\newtheorem{corollary}[theorem]{Corollary}

\usepackage{stmaryrd} 
\SetSymbolFont{stmry}{bold}{U}{stmry}{m}{n}
\DeclarePairedDelimiter\sem{\llbracket}{\rrbracket}
\DeclareMathOperator\supp{supp}
\DeclareMathOperator\Log{Log}

\usepackage{pict2e} 
\makeatletter
\newcommand{\bigcomp}{%
  \DOTSB
  \mathop{\vphantom{\sum}\mathpalette\bigcomp@\relax}%
  \slimits@
}
\newcommand{\bigcomp@}[2]{%
  \begingroup\m@th
  \sbox\z@{$#1\sum$}%
  \setlength{\unitlength}{0.9\dimexpr\ht\z@+\dp\z@}%
  \vcenter{\hbox{%
    \begin{picture}(1,1)
    \bigcomp@linethickness{#1}
    \put(0.5,0.5){\circle{1}}
    \end{picture}%
  }}%
  \endgroup
}
\newcommand{\bigcomp@linethickness}[1]{%
  \linethickness{%
      \ifx#1\displaystyle 2\fontdimen8\textfont\else
      \ifx#1\textstyle 1.65\fontdimen8\textfont\else
      \ifx#1\scriptstyle 1.65\fontdimen8\scriptfont\else
      1.65\fontdimen8\scriptscriptfont\fi\fi\fi 3
  }%
}
\makeatother

%% file: 1_introduction.tex
\section{Introduction}\label{sec:introduction}

In this paper we demonstrate the use of a new circuit connective called the \textit{control product}~\cite{SchoberWesley2026} to explicitly convert between two types of quantum adder circuits.
Quantum adders fall broadly into two categories: the Quantum Fourier Transform adder~\cite{draperAdditionQuantumComputer2000}, and the adders based on classical reversible approaches. 
The latter contains two main families, those based on the Ripple-Carry technique~\cite{vedralQuantumNetworksElementary1996,cuccaroNewQuantumRipplecarry2004,takahashiQuantumAdditionCircuits2009,remaudAncillaFreeQuantumAdder2025} and those based on the Carry-Lookahead technique~\cite{draperLogarithmicdepthQuantumCarrylookahead2004,takahashiQuantumAdditionCircuits2009,mogensenReversibleInPlaceCarryLookahead2019}.
While the connections between Ripple-Carry and Carry-Lookahead adders have been studied recently~\cite{remaudAncillaFreeQuantumAdder2025,remaudQuantumAddersStructural2025}, the connection between the classical reversible adders and the Quantum Fourier Transform adder remains unknown.
The difficulty in connecting the two is twofold. 
First, classical reversible adders are described using the gateset $\{X,\text{CNOT},\text{Toffoli}\}$, while the Quantum Fourier Transform adder is described used parameterized $Z$ rotation gates by small angles.
Second, the natural recursive structure of the Quantum Fourier Transform involves both a prefix and a suffix, while the classical reversible adders have a variety of internal structures like recursive ladders of CNOT or Toffoli gates.

We close this gap by explicitly transpiling the Quantum Fourier Transform adder~\cite{draperAdditionQuantumComputer2000} to one of the Ripple-Carry adders~\cite{takahashiQuantumAdditionCircuits2009}.
Along the way we produce many intermediate quantum adder circuits which range between the Quantum Fourier Transform adder and the Ripple-Carry adder, essentially producing a path through the space of quantum adders.
In~\cref{sec:qft_to_cla} we highlight one such intermediate quantum adder, a novel kind of Carry-Lookahead adder that is the first of its kind to use no ancilla qubits.
The technique used to perform this transpilation involves \textit{hierarchical quantum circuits}~\cite{SchoberWesley2026}, a slightly extended quantum circuit language whose main feature is the control product.
Hierarchical quantum circuits are more expressive than standard quantum circuits and allow for some circuit transformations that are difficult to find using standard quantum circuits.
Hierarchical quantum circuits represent unitary matrices, and therefore cannot represent anything that could not be compiled into a standard quantum circuit.
This distinguishes hierarchical quantum circuits from more exotic languages like the ZX-calculus~\cite{weteringZXcalculusWorkingQuantum2020}, which can represent a larger class of linear maps and transformations between them, but which also suffers from computational hardness problems for circuit synthesis~\cite{beaudrapCircuitExtractionZXdiagrams2022}.

The paper is structured as follows.
In~\cref{sec:introduction} we introduce the necessary background on hierarchical quantum circuits and quantum adders.
In~\cref{sec:qft_sandwich} we write the Quantum Fourier Transform adder~\cite{draperAdditionQuantumComputer2000} in recursive form to prepare for the transpilation.
In~\cref{sec:qft_to_cla} we convert the Quantum Fourier Transform Adder into a novel Carry-Lookahead adder by merging and canceling its powers to produce a quantum adder in prefix form.
Finally in~\cref{sec:cla_to_rc} we convert this new Carry-Lookahead Adder into the Ripple-Carry adder~\cite{takahashiQuantumAdditionCircuits2009}, completing the transpilation.

\subsection{Hierarchical Quantum Circuits}

\begin{figure}[t]
    \begin{subfigure}[b]{0.16\textwidth}
        \centering
        \circuitimg{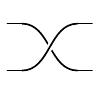}
        \caption{$\textsf{swap}: 2 \to 2$.}
    \end{subfigure}
    \begin{subfigure}[b]{0.16\textwidth}
        \centering
        \circuitimg{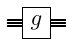}
        \caption{$g: n \to n$.}
    \end{subfigure}
    \begin{subfigure}[b]{0.16\textwidth}
        \centering
        \circuitimg{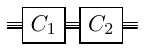}
        \caption{$C_1 \circ C_2$.}
    \end{subfigure}
    \begin{subfigure}[b]{0.16\textwidth}
        \centering
        \circuitimg{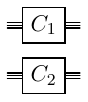}
        \caption{$C_1 \otimes C_2$.}
    \end{subfigure}
    \begin{subfigure}[b]{0.16\textwidth}
        \centering
        \circuitimg{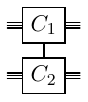}
        \caption{$C_1 \odot C_2$.}
    \end{subfigure}
    \begin{subfigure}[b]{0.16\textwidth}
        \centering
        \circuitimg{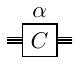}
        \caption{$C^\alpha$.}
    \end{subfigure}
    \vspace{-0.5em}
    \caption{The graphical language for hierarchical quantum circuit diagrams. Note that $C_1 \circ C_2$ is in circuit composition order, rather than matrix composition order.}
    \label{fig:hqc_syntax}
\end{figure}

In this section we summarize the hierarchical quantum circuit language introduced in~\cite{SchoberWesley2026}.
We begin with a gate set $\mathcal{G}$.
For each $g \in \mathcal{G}$, we write $g: n \to m$ to denote that $g$ has $n$ input wires and $m$ output wires.
We assume that all gates are unitary, so $g: n \to n$.
The hierarchical circuits generated by $\mathcal{G}$ are then defined inductively as follows.
\begin{itemize}
\item The the empty wire $\textsf{id}: 1 \to 1$ and the wire crossing $\textsf{swap}: 2 \to 2$ are hierarchical circuits.
\item If $g \in \mathcal{G}$, then $g$ is a hierarchical circuit.
\item If $C_1: n \to n$ and $C_2: n \to n$ are hierarchical circuits, then their \emph{sequential composition} denoted $C_1 \circ C_2: n \to n$ is a hierarchical circuit.
\item If $C_1: n \to n$ and $C_2: m \to m$ are hierarchical circuits, then their \emph{parallel composition} denoted $C_1 \otimes C_2: ( n + m ) \to ( n + m )$ is a hierarchical circuit.
\item If $C_1: n \to n$ and $C_2: m \to m$ are hierarchical circuits, then their \emph{control composition} denoted $C_1 \odot C_2: ( n + m ) \to ( n + m )$ is a hierarchical circuit.
\item If $C: n \to n$ is a hierarchical circuit and $\alpha \in \mathbb{R}$, then the \emph{circuit power} $C^\alpha$ is a hierarchical circuit.
\end{itemize}
The graphical language for these circuits can be found in~\cref{fig:hqc_syntax}.
As suggested by the graphical language, control composition is associative.
Repeated compositions $(\circ,\otimes,\odot)$ are denoted with big versions of the associated symbol $(\bigcomp,\bigotimes,\bigodot)$.

Now assume that each gate $g: n \to n$ has a semantic interpretation as some $( 2^n )$-dimensional unitary matrix $U_g$.
The semantic interpretation for hierarchical quantum circuits uses some standard definitions of matrix exponentials and logarithms, see e.g.~\cite{hallLieGroupsLie2015}.
First note that if $M$ is a matrix, then $\sum_{n=0}^\infty \frac{M^n}{n!}$ is a convergent series whose limit is denoted $\exp( M )$.
Moreover, if $U$ is a unitary matrix, then there exists a unique skew Hermitian matrix $\Log( U )$ with spectrum $i( -\pi, \pi ]$ such that $U = \exp( \Log( U ) )$.
Given these facts, the semantic interpretation of a hierarchical quantum circuit is defined inductively as follows, where $\sigma$ is the matrix such that $\sigma \ket{ij} = \ket{ji}$.
\begin{align*}
    \sem{\textsf{id}} &= I
    &
    \sem{\textsf{swap}} &= \sigma
    &
    \sem{g} &= U_g
\end{align*}
\begin{align*}
    \sem{C_1 \circ C_2} &= \sem{C_2}\sem{C_1}
    &
    \sem{C_1 \odot C_2} &= \exp( \Log\sem{C_1} \otimes \Log\sem{C_2} / (i\pi) )
    \\
    \sem{C_1 \otimes C_2} &= \sem{C_1} \otimes \sem{C_2}
    &
    \sem{C^\alpha} &= \exp( \alpha \Log\sem{C} )
\end{align*}
It should be noted that $\sem{U \odot (V \odot W)} = \sem{(U \odot V) \odot W)}$, so this interpretation is well-defined.
When the meaning is clear from context, we omit the $\sem{-}$.

It can be shown that every quantum circuit can be expressed by a hierarchical quantum circuit constructed from only the Pauli $Z$ gate $(\circuitimg{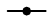})$, the Pauli $X$ gate $(\circuitimg{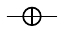})$, and the global scalar $-1$ $(\bullet)$.
For example, the Hadamard gate admits the following decomposition.
\begin{equation*}
    \circuitimg{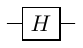}
    =
    \circuitimg{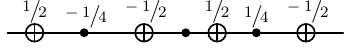}
\end{equation*}
The proof of single-qubit universality rests on the fact that parametrized $Z$-rotations, $X$-rotations, and global phases can be expressed in terms of circuit powers.
\begin{align*}
    \sem*{\circuitimg{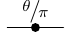}}
    =
    \begin{bsmallmatrix}
        1 & 0 \\
        0 & e^{i\theta}
    \end{bsmallmatrix}
    =
    Z( \theta )
    &
    &
    &
    &
    \sem*{\circuitimg{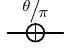}}
    =
    \tfrac{1}{2}
    \begin{bsmallmatrix}
        1 + e^{i\theta} & 1 - e^{i\theta} \\
        1 - e^{i\theta} & 1 + e^{i\theta}
    \end{bsmallmatrix}
    =
    HZ( \theta )H
    &
    &
    &
    &
    \sem*{\circuitimg{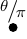}}
    =
    e^{i\theta}
\end{align*}
The multi-qubit case then follows from the fact that standard controlled unitaries  are recovered using control composition in the obvious way suggested by the notation.
\begin{align*}
     \sem*{\circuitimg{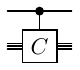}}
    =
    \sem*{\circuitimg{tikz-cache/section_1__definitions__basic_gates__z.pdf} \odot C}
    =
    \exp\left(
        \frac{ \Log ( Z ) \otimes \Log\sem{C} }{ i\pi }
        \right)
    =
    \exp\left( \begin{bsmallmatrix}
        0 & 0 \\
        0 & \Log\sem{C}
    \end{bsmallmatrix} \right)
    =
    \begin{bsmallmatrix}
        I & 0 \\
        0 & \sem{C}
    \end{bsmallmatrix}
\end{align*}
In~\cite{SchoberWesley2026}, a sound and complete equational theory was discovered for these gates.
This means that $\sem{C_1} = \sem{C_2}$ if and only if $C_1$ can be transformed into $C_2$ using the relations in the theory.

\subsection{Bit Order Convention}

All of the circuits in this paper will act on pairs of qubit registers denoted $\ket{ab}$, where $a = a_{n-1}a_{n-2}...a_1a_0$ and $b = b_{n-1}b_{n-2}...b_1b_0$ are bit-strings encoding integers.
For the sake of readability of the circuit diagrams, and consistency with the existing literature, the circuits in this paper will use the qubit order $\ket{(ab)} := \ket{a_0b_0a_1b_1...a_{n-1}b_{n-1}}$.
That is, the topmost pair of wires hold the least significant bits $a_0,b_0$ of $a$ and $b$, the bottommost pair of wires hold the most significant bits $a_{n-1}$ and $b_{n-1}$, and the bits of $a$ and $b$ have been interleaved.
We denote this reversed-and-interleaved bit ordering with parentheses $(ab) := a_0b_0a_1b_1...a_{n-1}b_{n-1}$.
A subscript $(ab)_{k:l} := a_kb_k...a_{l}b_{l}$ with $0 \leq k \leq l \leq n$ denotes the substring of $(ab)$ from index pair $k$ to index pair $l$.
Likewise subscripts on circuits $C_{k:l}$ denotes that $C$ is a circuit acting on the $l-k+1$ wire pairs $\ket{(ab)_{k:l}}$.
When a gate $g$ acts on a single qubit e.g. $a_j$, we will explicitly label it as $g_{a_j}$, even though this is a slight abuse of notation; $Z_{a_j} \otimes X_{b_j}$ is meant to be the subcircuit with $Z$ on qubit $a_j$, $X$ on qubit $b_j$, and nothing elsewhere.

We note that this bit ordering causes the QFT to overlap the $a$ register in the diagrams even though it does nothing there, as illustrated in~\cref{fig:draper_adder,fig:draper_adder_reordered}.

\subsection{Quantum Adders}

A \textit{quantum adder} is a quantum circuit with the following action on computational basis states.
\begin{equation*}
    \text{Adder} : \ket{a}\ket{b} \mapsto \ket{a}\ket{a+b}
\end{equation*}
The input registers $\ket{a}$ and $\ket{b}$ are computational basis states on $n$ qubits each, encoding integers $a = a_{n-1}a_{n-2}...a_1a_0$ and $b = b_{n-1}b_{n-2}...b_1b_0$.
The output register $\ket{a+b}$ contains their sum $s := (a+b)$ modulo $2^n$, hereafter denoted $s = s_{n-1}s_{n-2}...s_1s_0$.
There are many different quantum adders that vary in parameters such size, depth, and number of ancillas needed~\cite{vedralQuantumNetworksElementary1996,draperAdditionQuantumComputer2000,cuccaroNewQuantumRipplecarry2004,takahashiLinearsizeQuantumCircuit2005,takahashiQuantumAdditionCircuits2009,remaudAncillaFreeQuantumAdder2025,gidneyClassicalQuantumAdderConstant2025,remaudQuantumAddersStructural2025,vandaeleAsymptoticallyOptimalQuantum2026}. 
What is relevant to this paper is that they also differ in both native gateset and \textit{structure}, by which we mean the presence or absence of a prefix or suffix when written in recursive form.
We will focus on two quantum adders: one based on the Quantum Fourier Transform (QFT)~\cite{draperAdditionQuantumComputer2000}, and another based on the classically-inspired Ripple-Carry (RC) approach~\cite{takahashiQuantumAdditionCircuits2009}.

Draper~\cite{draperAdditionQuantumComputer2000} proposed the quantum adder shown in~\cref{fig:draper_adder}, where QFT denotes the Quantum Fourier Transform\footnote{We note that the QFT as defined by Draper differs from the typical presentation in that it does not swap the qubits at the end, and hence leaves them in reverse order.
Since the QFT is later followed by the inverse QFT in Draper's adder, the swaps can be neglected.
The swaps are inherited by the implementation of the TA circuit.
We follow Draper's convention.}, $\ket{\hat{b}} := (\text{QFT})\ket{b}$ denotes the Fourier transform of $\ket{b}$, and TA (Transform Addition) is a diagonal circuit consisting of only controlled $Z$ rotations with the action
\begin{equation*}
    \text{TA} : \ket{a}\ket{\hat{b}} \mapsto \ket{a}\ket{\hat{s}}.
\end{equation*}
The QFT-Adder works by first transforming the $\ket{b}$ register into the Fourier basis, then using the TA circuit to add $\ket{a}$ to the $\ket{\hat{b}}$ register, and finally un-transforming $\ket{\hat{s}}$ to $\ket{s}$ with the inverse QFT.
The QFT-Adder's native gateset is $\{H,CZ(2\pi/2^{k})\}$, where $CZ(2\pi/2^{k}) : k \in \mathbb{N}$ is a controlled, parametrized $Z$-rotation gate.
Structurally, the QFT-Adder is a sandwich, meaning its recursive implementation has both a prefix and a suffix.
\begin{figure}
    \centering
    \circuitimg{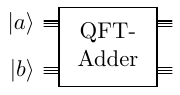}
    =
    \circuitimg{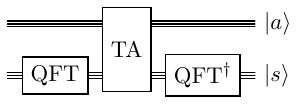}
    \caption{The Quantum Fourier Transform Adder as introduced in~\cite{draperAdditionQuantumComputer2000}.}
    \label{fig:draper_adder}
\end{figure}
\begin{figure}
    \centering
    \circuitimg{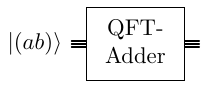}
    =
    \circuitimg{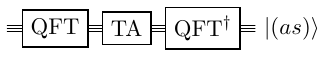}
    \caption{The QFT-Adder drawn with the bit order $(ab) = a_0b_0...a_{n-1}b_{n-1}$ used in this paper. The interleaving of $a$ and $b$ causes the QFT box to overlap the $a$ register, even though it does not act there. This circuit is equivalent to~\cref{fig:draper_adder} up to conjugation by swaps.}
    \label{fig:draper_adder_reordered}
\end{figure}

Takahashi et al.~\cite{takahashiQuantumAdditionCircuits2009} proposed the quantum adder shown in~\cref{fig:takahashi_adder}, based on the Ripple-Carry (RC) approach~\cite{vedralQuantumNetworksElementary1996,cuccaroNewQuantumRipplecarry2004} from classical reversible circuits.
This adder has seven stages which we draw as numbered boxes; we defer a description of their precise implementation to~\cref{sec:cla_to_rc}.
The RC-Adder works by first computing all of the carry bits $c_j$ needed for the basic grade school addition algorithm.
This is done by the first three stages. 
\begin{equation*}
    (1 \circ 2 \circ 3) : \ket{(ab)} \mapsto \ket{a_0b_0}\left(\bigotimes_{j=1}^{n-1}\ket{a_j + c_j}\ket{b_j + a_j}\right)
\end{equation*}
The remaining stages use the carry bits $c_j$ to compute $s_j = a_j + b_j + c_j$ in the $b$ register, and then uncompute the remaining junk in the $a$ register.
The native gate set of RC-Adder is the classical reversible gateset $\{X,\text{CNOT},\text{Toffoli}\}$.
Its individual stages can all be written recursively with either a prefix or a suffix, but the entire circuit has neither a single prefix nor a single suffix.
\begin{figure}
    \centering
    \circuitimg{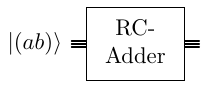}
    =
    \circuitimg{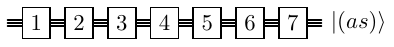}
    \caption{The Ripple-Carry Adder first introduced in~\cite{takahashiQuantumAdditionCircuits2009}. Its seven stages are defined in~\cref{sec:cla_to_rc}.}
    \label{fig:takahashi_adder}
\end{figure}
\begin{figure}
    \centering
    \circuitimg{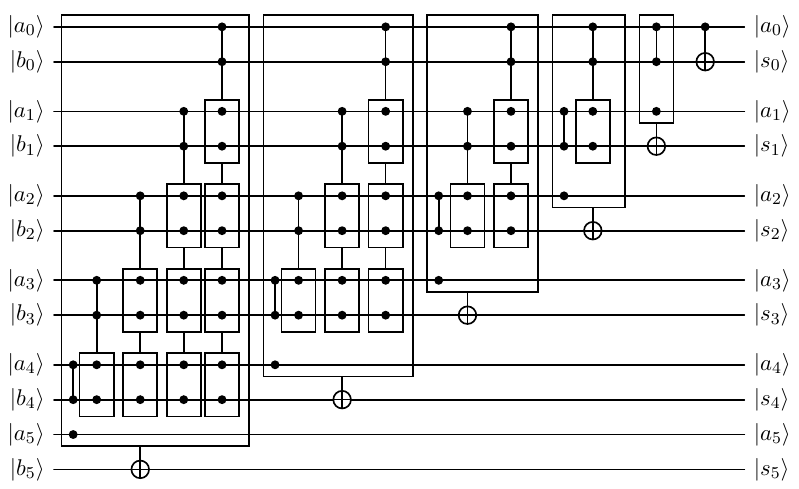}
    \caption{The Carry-Lookahead Adder described in~\cref{sec:qft_to_cla} for $n=6$. The $CZ = Z \odot Z$ gates identify \textit{generating bit pairs} whose addition will generate a carry, and the $Z \otimes Z$ boxes identify \textit{propagating bit pairs} whose addition will propagate a carry from above.}
    \label{fig:cl_adder_ex}
\end{figure}

In~\cref{sec:qft_to_cla} we produce a novel quantum adder as an intermediate transpilation stage that has a natural interpretation as a kind of Carry-Lookahead (CL) Adder. 
This CL-Adder works by computing the bits of $s$ one by one, from most significant to least significant.
It does so by controlling an $X$ gate using the subcircuit CLprefix as the control. 
CLprefix is a diagonal, Hermitian circuit whose $\pm 1$ eigenspaces distinguish pairs of bit-strings $a$ and $b$ based on whether their addition would cause a carry onto the next most significant bit or not.
In this way, CL-Adder computes the most significant bit $s_j$ without computing or storing any of the less significant carry bits, hence the name Carry-Lookahead.
\begin{equation*}
    (\text{CLprefix}_{0:j-1} \otimes Z_{a_j}) \odot X_{b_j} : \ket{(ab)} \mapsto \ket{(ab)_{0:j-1}a_js_j}
\end{equation*}
The native gate set for this adder involves hierarchical gates. 
CL-Adder can be compiled straightforwardly into the classical reversible gateset $\{X,\text{CNOT},\text{Toffoli}\}$ by repeatedly applying the identity $(Z_1 \otimes Z_2) \odot A = (Z_1 \odot A) \circ (Z_2 \odot A)$, where $A$ is any circuit; however this naive compilation results in an exponentially deep circuit. 
We remark that one expects to be able to find a depth-efficient implementation of CL-Adder, because many depth-efficient Carry-Lookahead Adders already exist in the literature~\cite{draperLogarithmicdepthQuantumCarrylookahead2004,takahashiQuantumAdditionCircuits2009,mogensenReversibleInPlaceCarryLookahead2019}, though notably none that use no ancilla qubits.
To the best of our knowledge this CL-Adder is the first quantum adder of its kind to not use any ancilla qubits.
When written in recursive form, it has a prefix structure.
An example of CL-Adder for $n=6$ is shown in~\cref{fig:cl_adder_ex}.

%% file: 2_qft_sandwich.tex
\section{Writing the QFT-Adder as a recursive sandwich}\label{sec:qft_sandwich}

Both the QFT and the TA circuits can be defined recursively as follows with a prefix structure, shown in~\cref{fig:qft_recursive_defn,fig:ta_recursive_defn}.
\begin{align*}
    &\text{QFTprefix}_{0:n-1} := 
    \bigcomp_{k=1}^{n}(Z_{b_{n-k}})^{2^{-k}}
    \\
    &\text{QFT}_{0:n-1} := 
    \begin{cases}
        H_{b_0} & n=1 \\
        \left( \text{QFTprefix}_{0:n-2}  \odot X_{b_{n-1}} \right) \circ \left( \text{QFT}_{0:n-2} \otimes H_{b_{n-1}}  \right) & n>1
    \end{cases}
    \\
    &\text{TAprefix}_{0:n-1} := 
    \bigcomp_{k=1}^{n}(Z_{a_{n-k}})^{2^{-k}}
    \\
    &\text{TA}_{0:n-1} := 
    \begin{cases}
        Z_{a_0} \odot Z_{b_0} & n=1 \\
        \left( \text{TAprefix}_{0:n-2} \odot Z_{b_{n-1}} \right) \circ \left( \text{TA}_{0:n-2} \otimes \left( Z_{a_{n-1}} \odot Z_{b_{n-1}} \right) \right) & n>1
    \end{cases}
\end{align*}
\begin{figure}
    \centering
    \circuitimg{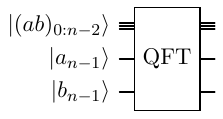}
    =
    \circuitimg{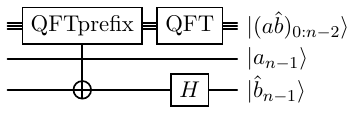}
    \caption{Recursive definition of the Quantum Fourier Transform circuit.}
    \label{fig:qft_recursive_defn}
\end{figure}
\begin{figure}
    \centering
    \circuitimg{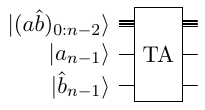}
    =
    \circuitimg{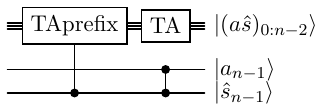}
    \caption{Recursive definition of the Transform Addition circuit.}
    \label{fig:ta_recursive_defn}
\end{figure}

By writing the QFT-Adder (\cref{fig:draper_adder_reordered}) using the recursive forms of the QFT and the TA, we note that $\text{TAprefix}_{0:n-2} \odot Z_{b_{n-1}}$ commutes with QFT$_{0:n-2}$ since it acts only on the $\ket{a}$ register and qubit $\ket{b_{n-1}}$, while QFT$_{0:n-2}$ only acts on $\ket{b_0...b_{n-2}}$.
Applying the identity $HZH = X$ on the $\ket{b_{n-1}}$ wire, we can then merge $(\text{QFTAprefix}_{0:n-2} \odot X_{b_{n-1}})$ and $(\text{TAprefix}_{0:n-2} \odot X_{b_{n-1}})$ into $\text{QFTAprefix}_{0:n-2} \odot X_{b_{n-1}}$, where $\text{QFTAprefix}$\footnote{Note the extra A in the name QFTAprefix compared to QFTprefix -- QFTAprefix is meant to stand for \textit{Quantum Fourier Transform Addition prefix}, since it is the combined prefix for both the Quantum Fourier Transform circuit and the Transform Addition circuit.} is the joint prefix of both the QFT and TA circuits.
\begin{align}\label{eqn:qfta_prefix}
     \text{QFTAprefix}_{0:n-1} := \text{QFTAprefix}_{0:n-1} \circ \text{TAprefix}_{0:n-1} = 
    \bigcomp_{k=1}^{n} \left( Z_{a_{n-k}}^{2^{-k}} \otimes Z_{b_{n-k}}^{2^{-k}} \right)
\end{align}
\begin{figure}
    \centering
    \circuitimg{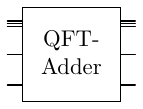}
    =
    \circuitimg{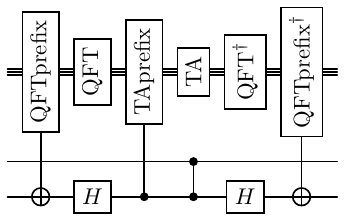}
    =
    \circuitimg{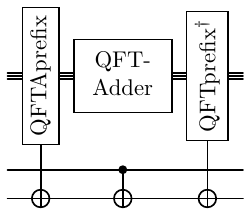}
    \caption{Recursive QFT-sandwich form of the Quantum Fourier Transform Adder.}
    \label{fig:qft_sandwich_small}
\end{figure}
This leads to the following recursive sandwich form of the QFT-Adder, shown in~\cref{fig:qft_sandwich_small}.
\begin{align}\label{eqn:qft_sandwich}
    \begin{split}
        \text{QFT-Adder}_{0:n-1} 
        &:= \left( \text{QFTAprefix}_{0:n-2} \odot X_{b_{n-1}} \right) \\
        &\circ \left( \text{QFT-Adder}_{0:n-2} \otimes (Z_{a_{n-1}} \odot X_{b_{n-1}}) \right) \\
        &\circ \left( (\text{QFTprefix}_{0:n-2})^\dagger \odot X_{b_{n-1}} \right)
    \end{split}
\end{align}

%% file: 3_qft_to_cla.tex
\section{Converting the QFT-Adder into a Carry-Lookahead Adder}\label{sec:qft_to_cla}

In this section we convert the QFT-Adder (\cref{eqn:qft_sandwich,fig:qft_sandwich_small}) into a type of Carry-Lookahead adder which we denote CL-Adder.
This is done by explicitly transpiling $\text{QFT-Adder}_{0:n-1}$ to $\text{CL-Adder}_{0:n-1}$ for all $n$ using their recursive definitions.
These two adders differ in both gateset and structure.
The gateset $\{H,CZ(2\pi/2^{k})\}$ of the QFT-Adder contains small-angle $Z$ rotations, encoded in hierarchical quantum circuits as circuit powers $Z^{(2^{-k})}$ with $k=1,...,n$.
In particular, the prefix circuit QFTAprefix defined in~\cref{eqn:qfta_prefix} contains two instances of each power $2^{-k}$.
The circuits raised to these powers will first be succesively merged in pairs, $2^{-k} + 2^{-k} = 2^{-k+1}$, until QFTAprefix contains exactly one subcircuit raised to each power $2^{-k}$ for $k = 0,1,...,n$.
The $X$-controlled suffix $(\text{QFTprefix}^\dagger)\odot X$ will then be pushed through the central Adder circuit, over to the left hand side of the circuit where it will cancel with the remaining subcircuits contained in QFTAprefix with a power different from 1 (i.e. $k\neq0$).
This results in CL-Adder, a quantum adder written over the power-free hierarchical gateset $\{Z,X\}$ with a prefix structure.

We begin by merging powers contained in $\text{QFTAprefix}$.
The merging process relies on two identities.
The first identity is as follows.
\begin{align}\label{eqn:merge_straight}
    &\sem{A},\sem{B} \text{ diagonal Hermitian}
    &\implies
    &
    &\circuitimg{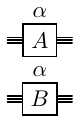}
    =
    \circuitimg{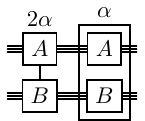}
\end{align}
The next identity will handle the bulk of the merging stage, and relies on the following notion of orthogonality. 
We denote the support of a matrix $A$ by $\supp( A )$.
Notice that $\supp( A )$ will be the complement of the zero eigenspace of $A$.
If $A$ and $B$ are matrices satisfying $\supp( \Log(A) ) \cap \supp( \Log(B) ) = \emptyset$, then we say that $A$ is \emph{orthogonal} to $B$ and write $A \perp B$.
Using this notion, the second identity is as follows.
\begin{align}\label{eqn:merge_zigzag}
    &\begin{array}{c}
        \sem{A},\sem{B},\sem{C},\sem{D} \\
        \text{diagonal Hermitian} \\
        \text{and} \\
        \sem{B} \perp \sem{D}
    \end{array}
    &\implies
    &
    &\circuitimg{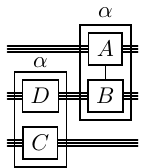}
    =
    \circuitimg{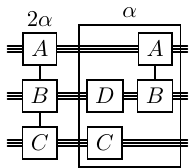}
\end{align}
It will also be convenient to define the following hierarchical circuits. 
These circuits act on $l-k+1$ pairs of wires $\ket{a_kb_k...a_{l}b_{l}}$. 
\begin{align*}
    &\text{Lookahead}_{k:l} 
    :=
    \begin{cases}
        Z_{a_k} \odot Z_{b_k} & k=l \\
        \text{Lookahead}_{k:l-1} \odot (Z_{a_{l}} \otimes Z_{b_{l}}) & k<l
    \end{cases}
    \\
    &\text{CLprefix}_{k:l} 
    :=
    \begin{cases}
        Z_{a_k} \odot Z_{b_k} & k=l \\
        \text{CLprefix}_{k+1:l} \circ \text{Lookahead}_{k:l} & k<l
    \end{cases}
\end{align*}
Expanding the recursion gives the following explicit forms.
\begin{align*}
    &\text{Lookahead}_{k:l} 
    := Z_{a_k} \odot Z_{b_k} \odot \left( \bigodot_{j=k+1}^{l} \left( Z_{a_j} \otimes Z_{b_j} \right) \right)
    \\
    &\text{CLprefix}_{k:l} 
    := \bigcomp_{j=0}^{l-k} \text{Lookahead}_{l-j:l}
\end{align*}
Examples of the $\text{Lookahead}$ and $\text{CLprefix}$ circuits on 5 wire pairs are shown in~\cref{fig:clprefixes_ex}.
\begin{figure}
    \centering
    \circuitimg{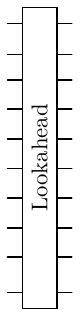}
    =
    \circuitimg{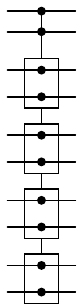}
    \hspace{1.5cm}
    \circuitimg{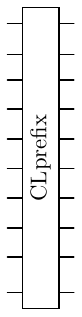}
    =
    \circuitimg{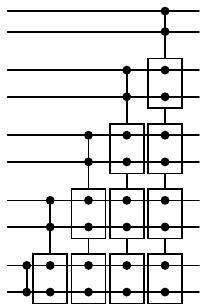}
    \caption{The Lookahead and CLprefix circuits on 5 qubit pairs.}
    \label{fig:clprefixes_ex}
\end{figure}

\subsection{Merging}
Now we merge the powers contained in the QFTAprefix circuit.
Merging a $\text{QFTAprefix}$ circuit acting on $n$ pairs of wires happens in a total of $n$ steps $l=1,...,n$, split across two stages. 
The first stage takes 2 steps.
In it, we apply~\cref{eqn:merge_straight} on all pairs of gates in $\text{QFTAprefix}$ with matching powers $2^{-k}$, twice. 
\begin{figure}
    \centering
    \circuitimg{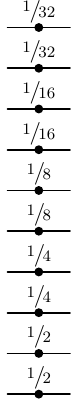}
    $\underset{(l=1)}{\overset{(\ref{eqn:merge_straight})}{=}}$
    \circuitimg{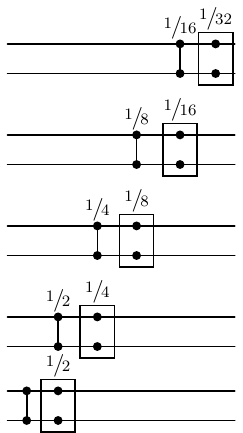}
    $\underset{(l=2)}{\overset{(\ref{eqn:merge_straight})}{=}}$
    \circuitimg{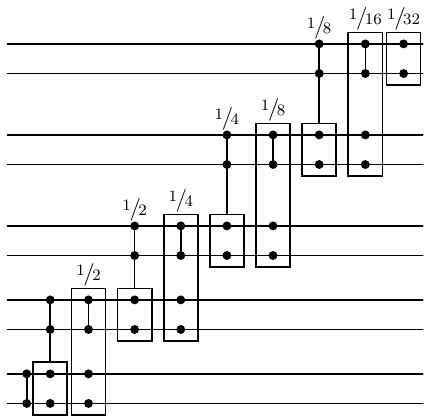}
    \caption{The two steps $l=1,2$ comprising the first merging stage for $\text{QFTAprefix}$ with $n = 5$.}
    \label{fig:mergestraight_ex}
\end{figure}
In step $l=1$ we apply~\cref{eqn:merge_straight}.
\begin{align*}
    &\text{QFTAprefix}_{0:n-1} \overset{(\ref{eqn:qfta_prefix})}{:=} 
    \bigcomp_{k=1}^{n} \left( Z_{a_{n-k}}^{2^{-k}} \otimes Z_{b_{n-k}}^{2^{-k}} \right)
    \\
    &\overset{(\ref{eqn:merge_straight})}{=}
    \bigcomp_{k=1}^{n} \left( \left( Z_{a_{n-k}} \odot Z_{b_{n-k}} \right)^{2^{-k+1}} \circ \left( Z_{a_{n-k}} \otimes Z_{b_{n-k}} \right)^{2^{-k}} \right).
\end{align*}
Or, regrouping terms by like power,
\begin{align*}
    \text{QFTAprefix}_{0:n-1} 
    &=
    \left( Z_{a_{n-1}} \odot Z_{b_{n-1}} \right) \circ
    \\
    &\circ \bigcomp_{k=1}^{n-1} \left( \left( Z_{a_{n-k-1}} \odot Z_{b_{n-k-1}} \right)^{2^{-k}} \otimes \left( Z_{a_{n-k}} \otimes Z_{b_{n-k}} \right)^{2^{-k}} \right) \circ
    \\
    &\circ \left( Z_{a_{0}} \otimes Z_{b_{0}} \right)^{2^{-n}}.
\end{align*}
In step $l=2$, we apply~\cref{eqn:merge_straight} again, with $A=Z_{a_{n-k-1}} \odot Z_{b_{n-k-1}}$ and $B=Z_{a_{n-k}} \otimes Z_{b_{n-k}}$ for each $k$. 
Note that $A \odot B = \text{Lookahead}_{n-k-1:n-k}$.
\begin{align*}
    &\overset{(\ref{eqn:merge_straight})}{=}
    \left( Z_{a_{n-1}} \odot Z_{b_{n-1}} \right) \circ
    \\
    &\circ \bigcomp_{k=1}^{n-1} \left( \left( \text{Lookahead}_{n-k-1:n-k} \right)^{2^{-k+1}} \circ \left( \left( Z_{a_{n-k-1}} \odot Z_{b_{n-k-1}} \right) \otimes Z_{a_{n-k}} \otimes Z_{b_{n-k}} \right)^{2^{-k}} \right) \circ
    \\
    &\circ \left( Z_{a_0} \otimes Z_{b_0} \right)^{2^{-n}}.
\end{align*}
Or, regrouping terms by like power,
\begin{align*}
    &=
    \left( Z_{a_{n-1}} \odot Z_{b_{n-1}} \right) \circ \text{Lookahead}_{n-2:n-1} \circ
    \\
    &\circ \bigcomp_{k=1}^{n-2} \left( \left( \left( Z_{a_{n-k-1}} \odot Z_{b_{n-k-1}} \right) \otimes Z_{a_{n-k}} \otimes Z_{b_{n-k}} \right)^{2^{-k}} \circ \left( \text{Lookahead}_{n-k-2:n-k-1} \right)^{2^{-k}} \right) \circ
    \\
    &\circ \left( \left( Z_{a_0} \odot Z_{b_0} \right) \otimes Z_{a_{1}} \otimes Z_{b_{1}} \right)^{2^{1-n}} \circ
    \\
    &\circ \left( Z_{a_0} \otimes Z_{b_0} \right)^{2^{-n}}.
\end{align*}
Substituting the definition of CLprefix in place of the remaining instances of $(\odot)$ gives the following form.
\begin{align*}
    &\overset{(\ref{eqn:merge_straight})}{=}
    \text{CLprefix}_{n-2:n-1} \circ 
    \\
    &\circ \bigcomp_{k=1}^{n-2} \left( \left( \text{CLprefix}_{n-k-1:n-k-1} \otimes Z_{a_{n-k}} \otimes Z_{b_{n-k}} \right)^{2^{-k}} \circ \left( \text{Lookahead}_{n-k-2:n-k-1} \right)^{2^{-k}} \right) \circ
    \\
    &\circ \left( \text{CLprefix}_{0:0} \otimes Z_{a_{1}} \otimes Z_{b_{1}} \right)^{2^{1-n}} \circ
    \\
    &\circ \left( Z_{a_0} \otimes Z_{b_0} \right)^{2^{-n}}.
\end{align*}
An example of these two steps with $n=5$ is shown in~\cref{fig:mergestraight_ex}.

The second stage takes $n-2$ steps.
At each step $l=3,...,n$ we apply~\cref{eqn:merge_zigzag} on all pairs of hierarchical gates with matching powers $2^{-k} < 1$.
This step can be repeated because the resulting circuit still contains pairs of hierarchical gates with matching powers $2^{-k}$ that meet the conditions of~\cref{eqn:merge_zigzag}.

More precisely, consider step $l > 2$. 
Assume that the result of the first $l-1$ merging steps has produced a circuit of the following form, where $\text{CLprefix}_{0:-1}$ is defined to be the empty circuit.
\begin{align}\label{eqn:merge_induction_hyp}
    \begin{split}
        &=
        \text{CLprefix}_{n-l+1:n-1} \circ
        \\
        &\circ \bigcomp_{k=1}^{n-l+1} \left( (\text{CLprefix}_{n-l+2-k:n-1-k} \otimes Z_{a_{n-k}} \otimes Z_{b_{n-k}} )^{2^{-k}} \circ \left( \text{Lookahead}_{n-l+1-k:n-1-k} \right)^{2^{-k}} \right) \circ
        \\
        &\circ \bigcomp_{k=n-l+2}^{n} \left( \text{CLprefix}_{0:n-1-k} \otimes Z_{a_{n-k}} \otimes Z_{b_{n-k}} \right)^{2^{-k}}
    \end{split}
\end{align}
Noting that 
\begin{equation*}
    \text{Lookahead}_{n-l+1-k:n-1-k} = \left( Z_{a_{n-l+1-k}} \odot Z_{b_{n-l+1-k}} \right) \odot \left( \bigodot_{j=1}^{l-2} Z_{a_{n-l+1-k+j}} \otimes Z_{b_{n-l+1-k+j}} \right),
\end{equation*}
then for each $k$ set
\begin{align*}
    &A=Z_{a_{n-l+1-k}} \odot Z_{b_{n-l+1-k}}
    &
    &B=\bigodot_{j=1}^{l-2} Z_{a_{n-l+1-k+j}} \otimes Z_{b_{n-l+1-k+j}}
    \\
    &D=\text{CLprefix}_{n-l+2-k:n-1-k}
    &
    &C=Z_{a_{n-k}} \otimes Z_{b_{n-k}}.
\end{align*}
Note that $A,B,C,D$ are all diagonal and Hermitian since they are all constructed only from the diagonal Hermitian gate $Z$ using the $\{\circ,\otimes,\odot\}$ connectives, which all map the set of diagonal Hermitian circuits onto itself.
Further note that $B \perp D$ since $\Log\sem*{B}$ is unsupported on every subspace of the form $\ket{11}_{a_{j}b_{j}} \otimes (-)$ which contain $\ket{11}$ on at least one wire pair $j$, while $\Log\sem*{D}$ is nonzero only on subspaces of this form, by construction.
Since $A,B,C,D$ are arranged as shown in the left-hand side of~\cref{eqn:merge_zigzag} and are diagonal Hermitian, and $B \perp D$, the conditions of~\cref{eqn:merge_zigzag} are satisfied and it can be applied.
Applying~\cref{eqn:merge_zigzag}, we have
\begin{align*}
    &=
    \text{CLprefix}_{n-l+1:n-1} \circ
    \\
    &\circ \bigcomp_{k=1}^{n-l+1} \left( \left( \text{Lookahead}_{n-l+1-k:n-k} \right)^{2^{-k+1}} \circ (\text{CLprefix}_{n-l+1-k:n-1-k} \otimes Z_{a_{n-k}} \otimes Z_{b_{n-k}} )^{2^{-k}} \right)\circ
    \\
    &\circ \bigcomp_{k=n-l+2}^{n} \left( \text{CLprefix}_{0:n-1-k} \otimes Z_{a_{n-k}} \otimes Z_{b_{n-k}} \right)^{2^{-k}}
\end{align*}
Rearranging terms to group those of like powers,
\begin{align*}
    &=
    \text{Clprefix}_{n-l:n-1} \circ 
    \\
    &\circ \bigcomp_{k=1}^{n-l} \left( (\text{CLprefix}_{n-l+1-k:n-1-k} \otimes Z_{a_{n-k}} \otimes Z_{b_{n-k}} )^{2^{-k}} \circ \left( \text{Lookahead}_{n-l-k:n-k-1} \right)^{2^{-k}} \right) \circ \\
    &\circ \bigcomp_{k=n-l+1}^{n} \left( \text{CLprefix}_{0:n-1-k} \otimes Z_{a_{n-k}} \otimes Z_{b_{n-k}} \right)^{2^{-k}}
\end{align*}
This circuit again has the form~\cref{eqn:merge_induction_hyp}, for the step $l+1$.
Since this subcircuit is present at step $l=3$, and the application of~\cref{eqn:merge_zigzag} at step $l$ reproduces it at step $l+1$, then by induction, merging will terminate after $l=n$ steps with the following circuit.
\begin{align*}
    &\text{QFTAprefix}_{0:n-1}
    \overset{(\ref{eqn:merge_straight},\ref{eqn:merge_zigzag})}{=}
    \text{CLprefix}_{0:n-1} \circ \left( \bigcomp_{k=1}^{n} \left( \text{CLprefix}_{0:n-1-k} \otimes Z_{a_{n-k}} \otimes Z_{b_{n-k}} \right)^{2^{-k}} \right)
\end{align*}
An example of this circuit for $n=5$ is shown in~\cref{fig:mergezigzag_ex}.
Circuit diagrams of all five steps of this example are shown in~\cref{apdx:merging_ex}.
\begin{figure}
    \centering
    \circuitimg{tikz-cache/section_3__examples__merge_straight__1.pdf}
    $\overset{(\ref{eqn:merge_straight},\ref{eqn:merge_zigzag})}{=}$
    \circuitimg{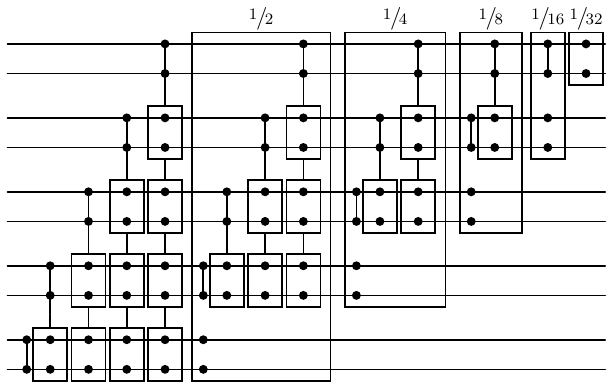}
    \caption{$\text{QFTAprefix}$ after all $n$ steps of both merging stages, with $n = 5$.}
    \label{fig:mergezigzag_ex}
\end{figure}

\subsection{Canceling}
The suffix $(\text{QFTprefix})^\dagger \odot X$ of QFT-Adder can now be pushed through the central Adder circuit, where it will cancel with the remaining hierarchical subcircuits with powers $2^{-k} < 1$ in the merged $\text{QFTAprefix}$ circuit.
The cancelation process relies on the following identity.
\begin{align}\label{eqn:xcopy}
    &A \text{ Hermitian}
    &\implies
    &
    &\circuitimg{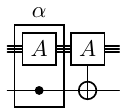}
    =
    \circuitimg{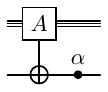}
\end{align}
Canceling all powers $2^{-k} < 1$ will produce the following form for an ancilla-free Carry-Lookahead adder.
\begin{align*}
    &\text{CL-Adder}_{0:n-1} :=
    \begin{cases}
        Z_{a_0} \odot X_{b_0} & n=1 \\
        \left( \text{CLprefix}_{0:n-2} \odot X_{b_{n-1}} \right) \circ \left( \text{CL-Adder}_{0:n-2} \otimes (Z_{a_{n-1}} \odot X_{b_{n-1}}) \right) & n>1
    \end{cases}
\end{align*}
Expanding the recursion, and merging the $X$ controls on the lowest wire, gives the following explicit form.
\begin{align}\label{eqn:cl_adder_defn}
    \text{CL-Adder}_{0:n-1} := \bigcomp_{k=1}^{n} \left( \left( \text{CLprefix}_{0:n-1-k} \otimes Z_{a_{n-k}} \right) \odot X_{b_{n-k}} \right)
\end{align}
\begin{figure}
    \centering
    \circuitimg{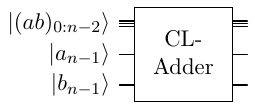}
    =
    \circuitimg{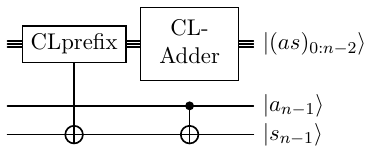}
    \caption{Recursive definition of the Carry-Lookahead Adder.}
    \label{fig:cla_defn}
\end{figure}

We now prove that CL-Adder is a quantum adder, as claimed, by showing that it is equivalent to QFT-Adder.
This is done by explicitly transpiling $\text{CL-Adder}_{0:n-1}$ to $\text{QFT-Adder}_{0:n-1}$ for all $n$.
\begin{proposition}\label{prop:cl=qft}
    $\text{CL-Adder}_{0:n-1} = \text{QFT-Adder}_{0:n-1}$ for all $n$.
\end{proposition}
\begin{proof}
    If $n=1$, we have $\text{CL-Adder}_{0:0} = Z \odot X = \text{CNOT} = \text{QFT-Adder}_{0:0}$ trivially.
    Now assume that $\text{CL-Adder}_{0:n-1} = \text{QFT-Adder}_{0:n-1}$ some $n \geq 1$.  By~\cref{eqn:qft_sandwich},
    \begin{align*}
        \text{QFT-Adder}_{0:n} 
        &\overset{(\ref{eqn:qft_sandwich})}{:=} \left( \text{QFTAprefix}_{0:n-1} \odot X_{b_{n}} \right) \\
        &\circ \left( \text{QFT-Adder}_{0:n-1} \otimes (Z_{a_{n}} \odot X_{b_{n}}) \right) \\
        &\circ \left( (\text{QFTprefix}_{0:n-1})^\dagger \odot X_{b_{n}} \right)
        \\
        &\overset{(\ref{eqn:merge_straight},\ref{eqn:merge_zigzag})}{=} \left( \left[ \text{CLprefix}_{0:n-1} \circ \left( \bigcomp_{k=1}^{n} \left( \text{CLprefix}_{0:n-1-k} \otimes Z_{a_{n-k}} \otimes Z_{b_{n-k}} \right)^{2^{-k}} \right) \right] \odot X_{b_{n}} \right) \\
        &\circ \left( \left[ \bigcomp_{k=1}^{n} \left( \left( \text{CLprefix}_{0:n-1-k} \otimes Z_{a_{n-k}} \right) \odot X_{b_{n-k}} \right) \right] \otimes (Z_{a_{n}} \odot X_{b_{n}}) \right) \\
        &\circ \left( \left[ \bigcomp_{k=1}^{n} (Z_{b_{n-k}})^{-2^{-k}} \right] \odot X_{b_{n}} \right)
    \end{align*}
    Pushing $(\text{QFTprefix}_{0:n-1})^\dagger \odot X_{b_{n}}$ from the right side of the circuit through the central $\text{CL-Adder}_{0:n-1} \otimes (Z_{a_{n}} \odot X_{b_{n}})$, we note two things.
    First, that both terms are controlled on the same $X_{b_{n}}$ gate, so their commutation relation is determined by the commutation relation of just $(\text{QFTprefix}_{0:n-1})^\dagger$ and $\text{CL-Adder}_{0:n-1}$.
    Second, each term $(Z_{b_{n-k}})^{-2^{-k}}$ contained in $(\text{QFTprefix}_{0:n-1})^\dagger$ fails to commute with exactly one term of $\text{CL-Adder}_{0:n-1}$, namely the $k$th term, due to the presence of the anti-commuting $X_{b_{n-k}}$.
    Applying~\cref{eqn:xcopy} with $A=(\text{CLprefix}_{0:n-1-k} \otimes Z_{a_{n-k}})$, we have the following.
    \begin{align*}
        &\overset{(\ref{eqn:xcopy})}{=} \left( \left[ \text{CLprefix}_{0:n-1} \circ \left( \bigcomp_{k=1}^{n} \left( \text{CLprefix}_{0:n-1-k} \otimes Z_{a_{n-k}} \otimes Z_{b_{n-k}} \right)^{2^{-k}} \right) \right] \odot X_{b_{n}} \right) \\
        &\circ \left( \left[ \bigcomp_{k=1}^{n} (\text{CLprefix}_{0:n-1-k} \otimes Z_{a_{n-k}} \otimes Z_{b_{n-k}})^{-2^{-k}} \right] \odot X_{b_{n}} \right) \\
        &\circ \left( \left[ \bigcomp_{k=1}^{n} \left( \left( \text{CLprefix}_{0:n-1-k} \otimes Z_{a_{n-k}} \right) \odot X_{b_{n-k}} \right) \right] \otimes (Z_{a_{n}} \odot X_{b_{n}}) \right)
    \end{align*}
    Canceling terms then recovers $\text{CL-Adder}_{0:n}$.
    \begin{align*}
        \text{QFT-Adder}_{0:n} 
        &= \left( \text{CLprefix}_{0:n-1} \odot X_{b_{n}} \right) \circ \\
        &\circ \left( \left[ \bigcomp_{k=1}^{n} \left( \left( \text{CLprefix}_{0:n-1-k} \otimes Z_{a_{n-k}} \right) \odot X_{b_{n-k}} \right) \right] \otimes (Z_{a_{n}} \odot X_{b_{n}}) \right) \\
        &= \left( \text{CLprefix}_{0:n-1} \odot X_{b_{n}} \right) \circ \left( \text{CL-Adder}_{0:n-1} \otimes (Z_{a_{n}} \odot X_{b_{n}}) \right) \\
        &= \text{CL-Adder}_{0:n}.
    \end{align*}
    Since $\text{QFT-Adder}_{0:0} = \text{CL-Adder}_{0:0}$ and $\text{QFT-Adder}_{0:n-1} = \text{CL-Adder}_{0:n-1} \implies \text{QFT-Adder}_{0:n} = \text{CL-Adder}_{0:n}$, then by induction $\text{QFT-Adder}_{0:n-1} = \text{CL-Adder}_{0:n-1}$ for all $n$.
\end{proof}

%% file: 4_cla_to_rc.tex
\section{Converting the Carry-Lookahead Adder into the Ripple-Carry Adder}\label{sec:cla_to_rc}
In this section we convert the CL-Adder (\cref{eqn:cl_adder_defn,fig:cla_defn}, full example in~\cref{fig:cl_adder_ex}) into the Ripple-Carry adder~\cite{takahashiQuantumAdditionCircuits2009}, which we denote RC-Adder. 
$\text{RC-Adder}_{0:n-1}$ has seven stages, which we draw as numbered boxes.
We remark that $\text{CNOT}(c,t) := Z_{c} \odot X_{t}$ and $\text{Toffoli}(c1,c2,t) := Z_{c1} \odot Z_{c2} \odot X_{t}$.
\begin{enumerate}
    \item Apply $\text{CNOT}(a_i,b_i)$ for $i=1$ to $n-1$.
    \item Apply $\text{CNOT}(a_{i},a_{i+1})$ for $i=n-2$ to 1.
    \item Apply $\text{Toffoli}(a_i,b_i,a_{i+1})$ for $i=0$ to $n-2$.
    \item Apply stage 1 again.
    \item For $i=n-2$ to $0$, apply a $X_{b_i}$ gate, followed by $\text{Toffoli}(a_i,b_i,a_{i+1})$, followed by another $X_{b_i}$ gate.
    \item Apply stage 2 in reverse order (for $i=1$ to $n-2$).
    \item Apply stage 1 again, and apply $\text{CNOT}(a_0,b_0)$.
\end{enumerate}
RC-Adder uses the gateset $\{X,\text{CNOT},\text{Toffoli}\}$ and can be written in recursive form as shown in~\cref{fig:rc_adder_defn}. 
\begin{figure}
    \centering
    \circuitimg{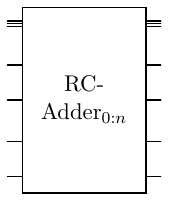}
    =
    \circuitimg{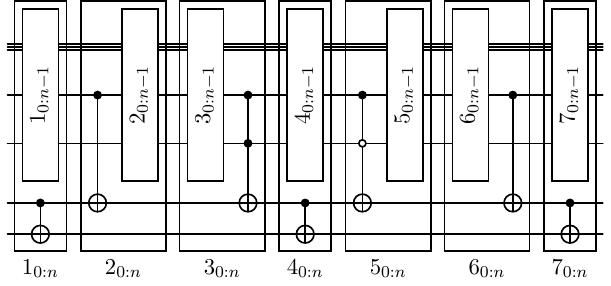}
    \caption{Recursive definition of the Ripple-Carry Adder.}
    \label{fig:rc_adder_defn}
\end{figure}
Each individual stage has a particularly simple structure, being either a prefix (stages 2 and 5), a suffix (stages 3 and 6), or either (stages 1, 4, and 7).
Since $XZX = -Z$, we will condense the three gates in stage 5 to the single gate $Z_{a_i} \odot (-Z)_{b_i} \odot X_{a_{i+1}}$, where $-Z$ is drawn as a white dot. 
This gate has the semantics suggested by the diagram notation, i.e. it is a Toffoli gate, but controlled by the state $\ket{10}_{a_ib_i}$ instead of $\ket{11}_{a_ib_i}$.
Hence each stage has a prefix/suffix containing exactly one gate.

To show that CL-Adder is equivalent to RC-Adder, we will again explicitly transpile $\text{CL-Adder}_{0:n-1}$ to $\text{RC-Adder}_{0:n-1}$ for all $n$.
This is done in the same way as in~\cref{sec:qft_to_cla}, by appealing to their recursive definitions.
We will push the prefix of $\text{CL-Adder}$ into the middle of $\text{RC-Adder}$ to reproduce the gates needed for the prefixes and suffixes of the seven stages shown in~\cref{fig:rc_adder_defn}.
In particular, we will push the circuit $\left( \text{CLprefix}_{0:n-1} \otimes Z_{a_{n}} \right) \odot X_{b_{n}}$ the first three numbered stages of the RC-Adder $1_{0:n}$, $2_{0:n}$, and $3_{0:n-1}$.
We denote the resulting circuit $C_{0:n}^{(j)}$ for $j = \{0,1,2,3\}$ and define it in the following way.
\begin{align*}
    &C_{0:n}^{(0)} := \left( \text{CLprefix}_{0:n-1} \otimes Z_{a_{n}} \right) \odot X_{b_{n}} \\
    &C_{0:n}^{(0)} \circ 1_{0:n} = 1_{0:n} \circ C_{0:n}^{(1)} \\
    &C_{0:n}^{(1)} \circ 2_{0:n} = 2_{0:n} \circ C_{0:n}^{(2)} \\
    &C_{0:n}^{(2)} \circ 3_{0:n-1} = 3_{0:n-1} \circ C_{0:n}^{(3)}
\end{align*}
In other words, $C_{0:n}^{(j)}$ is what remains of $C_{0:n}^{(0)}$ after being pushed through the first $j$ stages of $1_{0:n} \circ 2_{0:n} \circ 3_{0:n-1}$.
To prepare for the transpilation, we first compute $C_{0:n}^{(1)},C_{0:n}^{(2)}$, and $C_{0:n}^{(3)}$.

\subsection{Pushing through stage 1}
First we will compute $C_{0:n}^{(1)}$ using the two following identities.
~\cref{eqn:cnot_zz} is a corollary of~\cref{eqn:xcopy} when $A = Z$ and $\alpha = 1$, though it is better known as one of the standard `Pauli pushing' identities.
~\cref{eqn:cnot_cz} describes how $CZ = Z \odot Z$ changes when pushed through a CNOT.
\begin{center}
    \begin{minipage}{0.345\textwidth}
        \begin{align}\label{eqn:cnot_zz}
        \circuitimg{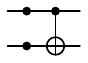}
        &=
        \circuitimg{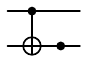}
    \end{align}
    \end{minipage}
    \hspace{0.2\textwidth}%
    \begin{minipage}{0.345\textwidth}
        \begin{align}\label{eqn:cnot_cz}
        \circuitimg{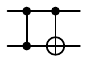}
        &=
        \circuitimg{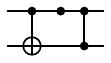}
    \end{align}
    \end{minipage}
\end{center}

Now consider how $C_{0:n}^{(0)}$ changes when pushed through $1_{0:n}$. 
Since $1_{0:n}$ is a vertical stack of $\text{CNOT}(a_j,b_j)$ gates (with one missing on the top wire pair $j=0$), we need only consider how $C_{0:n}^{(0)}$ changes by the action of CNOT on pairs of its wires $a_j,b_j$.
Since $\text{CLprefix}_{0:n-1}$ is comprised of $\text{Lookahead}$ gates, which are in turn comprised of $(Z_{a_j} \otimes Z_{b_j})$ and $(Z_{a_j} \odot Z_{b_j})$ pairs joined by $\odot$, it suffices to consider how these gates change when pushed through CNOT.

\begin{figure}
    \centering
    \circuitimg{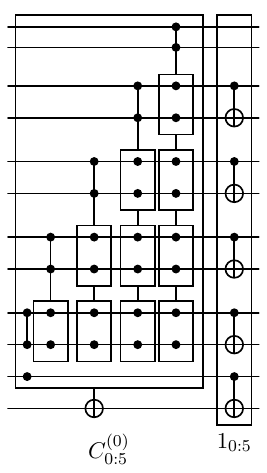}
    =
    \circuitimg{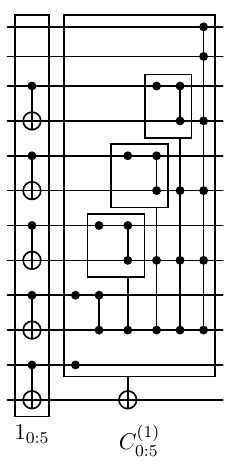}
    =
    \circuitimg{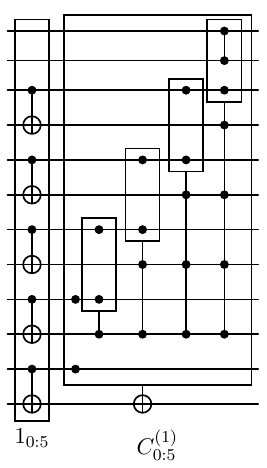}
    \caption{Pushing $C_{0:5}^{(0)}$ through $1_{0:5}$ results in $C_{0:5}^{(1)}$.}
    \label{fig:345_push_1}
\end{figure}
Each such wire pair has the form of either~\cref{eqn:cnot_zz} or~\cref{eqn:cnot_cz}.
Applying~\cref{eqn:cnot_zz} or~\cref{eqn:cnot_cz} on every pair of wires in $\text{Lookahead}$ gives the following, where $L(1,2) := Z_1 \circ (Z_1 \odot Z_2)$ is shorthand for the right-hand side of~\cref{eqn:cnot_cz} without the CNOT (so named because its circuit is shaped like an upside-down L).
\begin{align*}
    &\text{Lookahead}_{k:n-1} \circ 1_{0:n} = 
    \begin{cases}
        1_{0:n} \circ \left( Z_{a_0} \odot Z_{b_0} \odot \left( \bigodot_{j=1}^{n-1} Z_{b_j} \right) \right) & k=0 \\
        1_{0:n} \circ \left( L(a_k,b_k) \odot \left( \bigodot_{j=k+1}^{n-1} Z_{b_j} \right) \right) & k>0
    \end{cases}
\end{align*}
Pushing $C_{0:n}^{(0)}$ through $1_{0:n}$ therefore gives the following expressions for $C_{0:n}^{(1)}$, where in the last step we have only regrouped terms.
An example of this stage is shown in~\cref{fig:345_push_1}.
\begin{align}\label{eqn:push_1}
    \begin{split}
        C_{0:n}^{(1)} &= (1_{0:n})^\dagger \circ C_{0:n}^{(0)} \circ 1_{0:n} \\
        &= (1_{0:n})^\dagger \circ \left( \left[ \left( \bigcomp_{k=1}^{n} \text{Lookahead}_{n-k:n-1} \right) \otimes Z_{a_{n}} \right] \odot X_{b_{n}} \right) \circ 1_{0:n} \\
        &\overset{(\ref{eqn:cnot_zz},\ref{eqn:cnot_cz})}{=} \Bigg[  
        L(a_{n-1},b_{n-1}) \circ \bigcomp_{k=2}^{n-1} \left( L(a_{n-k},b_{n-k}) \odot \bigodot_{j=1}^{k-1} Z_{b_{n-k+j}} \right) \circ \\
        &\hspace{1cm} \circ \left( CZ(a_0,b_0) \odot \bigodot_{j=1}^{n-1} Z_{b_j} \right) \otimes Z_{a_n} \Bigg] \odot X_{b_n} \\
        &= \Bigg[ (Z_{a_{n-1}} \otimes Z_{a_{n}}) \circ \bigcomp_{k=2}^{n-1} \left( \left( Z_{a_{n-k}} \otimes Z_{a_{n-k+1}} \right) \odot \bigodot_{j=1}^{k-1} Z_{b_{n-k+j}} \right) \circ \\
        &\hspace{1cm} \circ \left( \left( CZ(a_0,b_0) \otimes Z_{a_1} \right) \odot \bigodot_{j=1}^{n-1} Z_{b_j} \right) \Bigg] \odot X_{b_n}
    \end{split}
\end{align}

\subsection{Pushing through stage 2}

\begin{figure}
    \centering
    \circuitimg{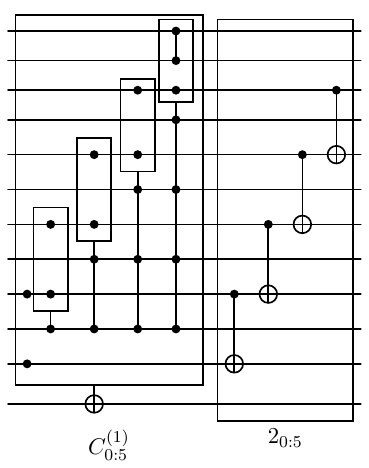}
    =
    \circuitimg{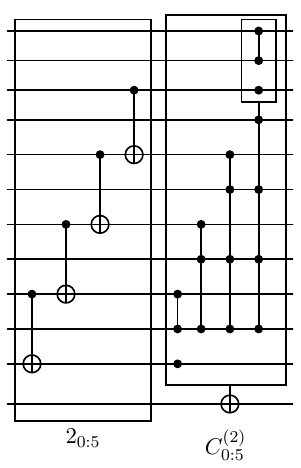}
    \caption{Pushing $C_{0:5}^{(1)}$ through $2_{0:5}$ results in $C_{0:5}^{(2)}$.}
    \label{fig:345_push_2}
\end{figure}
Next we compute $C_{0:n}^{(2)}$ by pushing $C_{0:n}^{(1)}$ through $2_{0:n}$. 
$2_{0:n}$ consists of an upward ladder of CNOT gates acting exclusively on the $a$ register; namely, on pairs $(a_j,a_{j+1})$ for $j = n-1,...,1$.
Consider the first of these gates.
$\text{CNOT}(a_{n-1},a_{n}) = Z_{a_{n-1}} \odot X_{a_{n}}$ does not commute with only one gate in $C_{0:n}^{(1)}$, namely a single $Z_{a_n}$ gate.
As this gate comes with a matching partner $Z_{a_{n-1}}$, we can apply~\cref{eqn:cnot_zz} to push the pair $Z_{a_{n-1}} \otimes Z_{a_{n}}$ through this CNOT, resulting in the $Z_{a_{n-1}}$ gate being deleted.
Now consider the next CNOT in the ladder, $\text{CNOT}(a_{n-2},a_{n-1}) = Z_{a_{n-2}} \odot X_{a_{n-1}}$.
Once again this does not commute with only one gate in what remains of $C_{0:n}^{(1)}$, a single $Z_{a_{n-1}}$ gate, which again has a partner $Z_{a_{n-2}}$ contained in the same controlled subcircuit $(Z_{a_{n-2}} \otimes Z_{a_{n-1}}) \odot Z_{b_{n-1}}$.
Hence we can apply~\cref{eqn:cnot_zz} again in this controlled subcircuit, deleting $Z_{a_{n-2}}$. 
We can then repeat this process on every controlled $Z_{a_{j}} \otimes Z_{a_{j+1}}$ subcircuit, deleting $Z_{a_{j}}$ gates for $j = n-1,...,1$ until $C_{0:n}^{(1)}$ has been pushed entirely through $2_{0:n}$, resulting in $C_{0:n}^{(2)}$.
An example of this stage is shown in~\cref{fig:345_push_2}.
\begin{align}\label{eqn:push_2}
    \begin{split}
        C_{0:n}^{(2)} &= (2_{0:n})^\dagger \circ C_{0:n}^{(1)} \circ 2_{0:n} \\
        &= \Bigg[ Z_{a_{n}} \circ \bigcomp_{k=2}^{n-1} \left( Z_{a_{n-k+1}} \odot \bigodot_{j=1}^{k-1} Z_{b_{n-k+j}} \right) \circ \\
        &\hspace{1cm} \circ \left( \left( CZ(a_0,b_0) \otimes Z_{a_1} \right) \odot \bigodot_{j=1}^{n-1} Z_{b_j} \right) \Bigg] \odot X_{b_n}
    \end{split}
\end{align}

\subsection{Pushing through stage 3}

Next we compute $C_{0:n}^{(3)}$ by pushing $C_{0:n}^{(2)}$ through $3_{0:n-1}$. 
This stage consists of a repeated application of the following identity, which is a corollary of~\cref{eqn:xcopy} when $A = Z \odot Z$ and $\alpha = 1$.
\begin{align}\label{eqn:toffoli_cz_z}
    \circuitimg{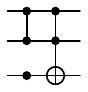}
    =
    \circuitimg{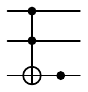}
\end{align}

$3_{0:n-1}$ consists of a downward ladder of Toffoli gates.
As in stage 2, consider what happens when we push $C_{0:n}^{(2)}$ through the first of these gates, $\text{Toffoli}(a_0,b_0,a_1) = Z_{a_0} \odot Z_{b_0} \odot X_{a_1}$.
Again there is only one gate inside $C_{0:n}^{(2)}$ that does not commute with $\text{Toffoli}(a_0,b_0,a_1)$, namely a single $Z_{a_1}$ contained in a controlled subcircuit $Z_{a_0} \odot Z_{b_0} \otimes Z_{a_1}$ at the right edge of the $C_{0:n}^{(2)}$ circuit.
Applying~\cref{eqn:toffoli_cz_z}, we find that the $Z_{a_0} \odot Z_{b_0}$ part of the controlled subcircuit is simply deleted when pushed past $\text{Toffoli}(a_0,b_0,a_1)$.
Grouping the remaining controls together, we find that this causes another instance of a controlled subcircuit $Z_{a_1} \odot Z_{b_1} \otimes Z_{b_2}$ to appear, on the next wire pair down.
An example of this single step is shown in~\cref{fig:345_push_3_single_step}.
\begin{figure}
    \centering
    \circuitimg{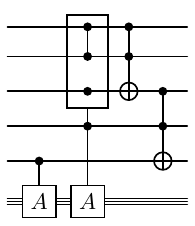}
    =
    \circuitimg{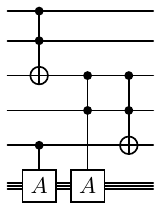}
    =
    \circuitimg{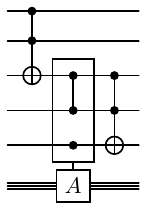}
    \caption{Pushing the rightmost controlled box contained in $C_{0:n}^{(2)}$ through a single Toffoli gate in $3_{0:n-1}$. $A$ is any circuit.}
    \label{fig:345_push_3_single_step}
\end{figure}
\begin{figure}
    \centering
    \circuitimg{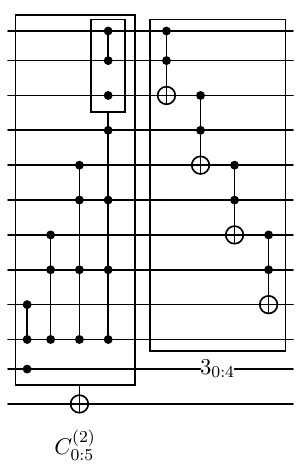}
    =
    \circuitimg{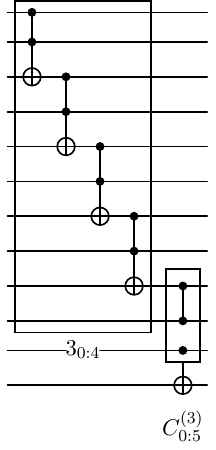}
    \caption{Pushing $C_{0:5}^{(2)}$ through $3_{0:4}$ results in $C_{0:5}^{(3)}$.}
    \label{fig:345_push_3}
\end{figure}
This process repeats with every Toffoli gate in $3_{0:n-1}$.
As $C_{0:n}^{(2)}$ pushes through the Toffoli ladder $3_{0:n-1}$, it collapses until only a single instance of the subcircuit $Z_{a_{n-1}} \odot Z_{b_{n-1}} \otimes Z_{a_n}$ remains.
This leads to the following form for $C_{0:n}^{(3)}$.\footnote{We now see that $C_{0:n}^{(3)}$ acts nontrivially on only $(ab)_{n-1:n}$, so more precisely it could be called $C_{n-1:n}^{(3)}$.}
An example of this pushing stage is shown in~\cref{fig:345_push_3}.
\begin{align}\label{eqn:push_3}
    C_{0:n}^{(3)} &= (3_{0:n-1})^\dagger \circ C_{0:n}^{(2)} \circ 3_{0:n-1} = \left[ \left( Z_{a_{n-1}} \odot Z_{b_{n-1}} \right) \otimes Z_{a_n} \right] \odot X_{b_n}
\end{align}

Now we are ready to transpile.
\begin{proposition}\label{prop:rc=cl}
    $\text{RC-Adder}_{0:n-1} = \text{CL-Adder}_{0:n-1}$ for all $n$.
\end{proposition}
\begin{proof}
    If $n=1$, the only stage that contributes a gate is $7$, and we have $\text{RC-Adder}_{0:0} = \text{CNOT} = Z \odot X = \text{CL-Adder}_{0:0}$ trivially.

    \begin{figure}
        \centering
        \circuitimg{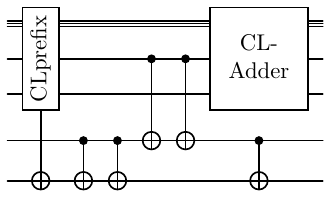}
        =
        \circuitimg{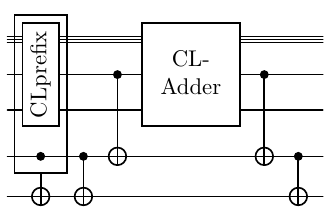}
        \caption{Preparing CNOTs to merge into stages 1,2,6, and 7 of the CL-Adder.}
        \label{fig:1267_create_cnots}
    \end{figure}
    \begin{figure}
        \centering
        \circuitimg{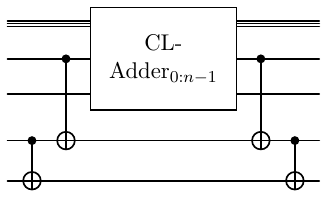}
        =
        \circuitimg{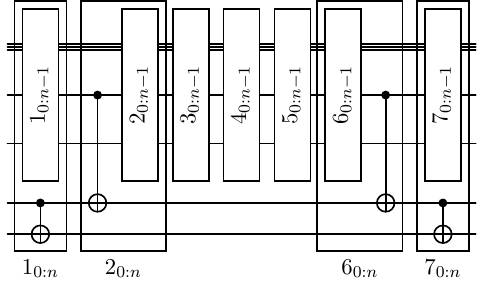}
        \caption{Merging CNOTs into stages 1,2,6, and 7 of CL-Adder. Stages 3,4, and 5 are still missing their respective gates.}
        \label{fig:1267_combine}
    \end{figure}
    Now assume that $\text{RC-Adder}_{0:n-1} = \text{CL-Adder}_{0:n-1}$ for some $n \geq 1$. 
    Plugging in the recursive definition of $\text{CL-Adder}_{0:n}$ and inserting two pairs of (self-inverse) CNOT gates, we note that $\text{CNOT}(a_{n-1},a_{n}) = Z_{a_{n-1}} \odot X_{a_{n}}$ commutes with $\text{CL-Adder}_{0:n-1}$. 
    This is most easily seen from the recursive definition of $\text{CL-Adder}$. 
    $\text{CL-Adder}_{0:n-1}$ and $\text{CNOT}(a_{n-1},a_{n})$ overlap only on wire $a_{n-1}$, where both share a $Z$ gate.
    Since $Z$ commutes with itself, it then follows that $\text{CNOT}(a_{n-1},a_{n})$ commutes with $\text{CL-Adder}_{0:n-1}$.
    Hence we can freely push $\text{CNOT}(a_{n-1},a_{n})$ through $\text{CL-Adder}_{0:n-1}$.
    Grouping one of the $\text{CNOT}(a_{n},b_{n})$ gates with $\text{CLprefix}_{0:n} \odot X_{b_n}$ then leads to the following form, shown in~\cref{fig:1267_create_cnots}.
    \begin{align*}
        \text{CL-Adder}_{0:n} &:= \left( \text{CLprefix}_{0:n-1} \odot X_{b_{n}} \right) \circ \left( \text{CL-Adder}_{0:n-1} \otimes (Z_{a_{n}} \odot X_{b_{n}}) \right) \\
        &= \left( \text{CLprefix}_{0:n-1} \odot X_{b_{n}} \right) \circ \\
        &\hspace{1cm} \circ \left( Z_{a_{n}} \odot X_{b_{n}} \right) \circ \left( Z_{a_{n}} \odot X_{b_{n}} \right) \circ \\
        &\hspace{1cm} \circ \left( Z_{a_{n-1}} \odot X_{a_{n}} \right) \circ \left( Z_{a_{n-1}} \odot X_{a_{n}} \right) \circ \\
        &\hspace{1cm} \circ \left( \text{CL-Adder}_{0:n-1} \otimes (Z_{a_{n}} \odot X_{b_{n}}) \right) \\
        &= \left( \left( \text{CLprefix}_{0:n-1} \otimes Z_{a_{n}} \right) \odot X_{b_{n}} \right) \circ \\
        &\hspace{1cm} \circ \left( Z_{a_{n}} \odot X_{b_{n}} \right) \circ \left( Z_{a_{n-1}} \odot X_{a_{n}} \right) \circ \\
        &\hspace{1cm} \circ \text{CL-Adder}_{0:n-1} \ \circ \\
        &\hspace{1cm} \circ \left( Z_{a_{n-1}} \odot X_{a_{n}} \right) \circ \left( Z_{a_{n}} \odot X_{b_{n}} \right)
    \end{align*}
    Now we apply the inductive hypothesis to replace $\text{CL-Adder}_{0:n-1}$ with $\text{RC-Adder}_{0:n-1}$.
    Inserting the recursive definition of $\text{RC-Adder}_{0:n-1}$, we note that $\text{CNOT}(a_{n-1},a_{n})$ commutes with $1_{0:n-1}$ and $7_{0:n-1}$ for the same reason as above.
    Hence we can push each $\text{CNOT}(a_{n-1},a_{n})$ on the left and right side through $1_{0:n-1}$ and $7_{0:n-1}$.
    The four CNOTs can then be grouped with $1_{0:n-1},2_{0:n-1},6_{0:n-1}$ and $7_{0:n-1}$ to form $1_{0:n},2_{0:n},6_{0:n}$ and $7_{0:n}$. 
    This is shown in~\cref{fig:1267_combine}.
    \begin{align*}
        \text{CL-Adder}_{0:n} &= C_{0:n}^{(0)} \circ 1_{0:n} \circ 2_{0:n} \circ 3_{0:n-1} \circ 4_{0:n-1} \circ 5_{0:n-1} \circ 6_{0:n} \circ 7_{0:n}
    \end{align*}
    Now we apply~\cref{eqn:push_1,eqn:push_2,eqn:push_3} to push $C_{0:n}^{(0)}$ through the first three stages.
    \begin{align*}
        \text{CL-Adder}_{0:n} &= \left( C_{0:n}^{(0)} \right) \circ 1_{0:n} \circ 2_{0:n} \circ 3_{0:n-1} \circ 4_{0:n-1} \circ 5_{0:n-1} \circ 6_{0:n} \circ 7_{0:n} \\
        &\overset{(\ref{eqn:push_1})}{=} 1_{0:n} \circ \left( C_{0:n}^{(1)} \right) \circ 2_{0:n} \circ 3_{0:n-1} \circ 4_{0:n-1} \circ 5_{0:n-1} \circ 6_{0:n} \circ 7_{0:n} \\
        &\overset{(\ref{eqn:push_2})}{=} 1_{0:n} \circ 2_{0:n} \circ \left( C_{0:n}^{(2)} \right) \circ 3_{0:n-1} \circ 4_{0:n-1} \circ 5_{0:n-1} \circ 6_{0:n} \circ 7_{0:n} \\
        &\overset{(\ref{eqn:push_3})}{=} 1_{0:n} \circ 2_{0:n} \circ 3_{0:n-1} \circ \left( C_{0:n}^{(3)} \right) \circ 4_{0:n-1} \circ 5_{0:n-1} \circ 6_{0:n} \circ 7_{0:n} \\
    \end{align*}
    Finally, we apply the following identity, which follows from~\cref{eqn:toffoli_cz_z} along with the usual Pauli anticommutation identity $\circuitimg{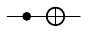} = \circuitimg{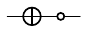}$.
    \begin{align}\label{eqn:push_4}
        \circuitimg{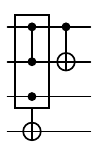}
        =
        \circuitimg{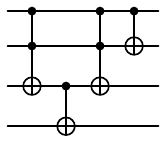}
        =
        \circuitimg{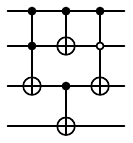}
    \end{align}
    Note that the three gates $\text{Toffoli}(a_{n-1},b_{n-1},a_n)$, $\text{CNOT}(a_n,b_n)$, and $Z_{a_{n-1}} \odot (-Z)_{a_{n-1}} \odot X_{a_{n}}$ in~\cref{eqn:push_4} are exactly the three missing gates in $3_{0:n}$, $4_{0:n}$, and $5_{0:n}$ respectively.
    Hence we can merge them to recover $3_{0:n}$, $4_{0:n}$, and $5_{0:n}$.
    This is shown in~\cref{fig:345_push_4}.
    \begin{align*}
        3_{0:n-1} \circ C_{0:n}^{(3)} \circ 4_{0:n-1} \circ 5_{0:n-1} = 3_{0:n} \circ 4_{0:n} \circ 5_{0:n}
    \end{align*}
    This finishes the transpilation.
    \begin{align*}
        \text{CL-Adder}_{0:n} &= 1_{0:n} \circ 2_{0:n} \circ 3_{0:n} \circ 4_{0:n} \circ 5_{0:n} \circ 6_{0:n} \circ 7_{0:n} = \text{RC-Adder}_{0:n}
    \end{align*}
    Since $\text{RC-Adder}_1 = \text{CL-Adder}_1$ and $\text{RC-Adder}_{0:n-1} = \text{CL-Adder}_{0:n-1} \implies \text{RC-Adder}_{0:n} = \text{CL-Adder}_{0:n}$, then by induction we have $\text{RC-Adder}_{0:n-1} = \text{CL-Adder}_{0:n-1}$ for all $n$, which complete the proof.
\end{proof}

\begin{figure}
    \centering
    \circuitimg{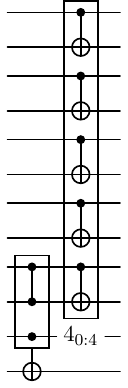}
    =
    \circuitimg{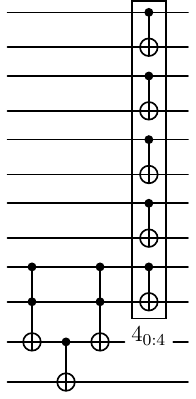}
    =
    \circuitimg{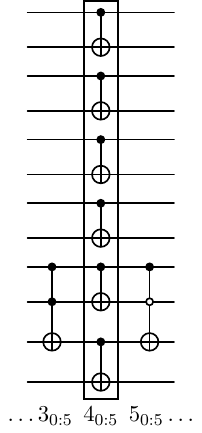}
    \caption{Pushing $C_{0:5}^{(3)}$ through $4_{0:4}$ results in the three missing gates to recover stages 3, 4, and 5 of Takahashi ripple-carry adder~\cite{takahashiQuantumAdditionCircuits2009}.}
    \label{fig:345_push_4}
\end{figure}

\begin{corollary}\label{prop:rc=qft}
    $\text{RC-Adder}_{0:n-1} = \text{QFT-Adder}_{n-1}$ for all $n$.
\end{corollary}
\begin{proof}
    Follows from~\cref{prop:cl=qft,prop:rc=cl}.
\end{proof}

%% file: 5_conclusion.tex
\section{Conclusion}\label{sec:conclusion}

We have used the technique of hierarchical quantum circuits to study the structural properties of various types of quantum adders.
In particular, we converted the Quantum Fourier Transform adder~\cite{draperAdditionQuantumComputer2000} into a quantum adder~\cite{takahashiQuantumAdditionCircuits2009} based on the Ripple-Carry technique from classical reversible logic.
This conversion took the form of an explicit gate-level transpilation, producing a path through the space of quantum adders that uses only local circuit equivalences at each step of the path.
We highlighted one intermediate quantum adder along this path that has a natural interpretation as an ancilla-free Carry-Lookahead adder.

%% file: appendix.tex
\section{Example of merging QFTAprefix for $n=5$}\label{apdx:merging_ex}

We show an explicit example of all 5 steps comprising the two stages of the merging procedure described in~\cref{sec:qft_to_cla}, for $n=5$.

Stage 1 of merging consists of steps $l=1$ and $l=2$. 
In these steps we apply~\cref{eqn:merge_straight} to all pairs of subcircuits with matching powers $2^{-k} : k > 0$.
These steps are also shown in~\cref{fig:mergestraight_ex}.
\begin{align*}
  \circuitimg{tikz-cache/section_3__examples__merge_straight__1.pdf}
  &\underset{(l=1)}{\overset{(\ref{eqn:merge_straight})}{=}}
  \circuitimg{tikz-cache/section_3__examples__merge_straight__2.pdf}
  \underset{(l=2)}{\overset{(\ref{eqn:merge_straight})}{=}}
  \circuitimg{tikz-cache/section_3__examples__merge_straight__3.pdf}
\end{align*}


Stage 2 of merging consists of steps $l=3,4,5$. 
In these steps we apply~\cref{eqn:merge_zigzag} to all pairs of subcircuits with matching powers $2^{-k} : k > 0$.
This results in the circuit shown in~\cref{fig:mergezigzag_ex}.
\begin{align*}
  &\underset{(l=3)}{\overset{(\ref{eqn:merge_zigzag})}{=}}
  \circuitimg{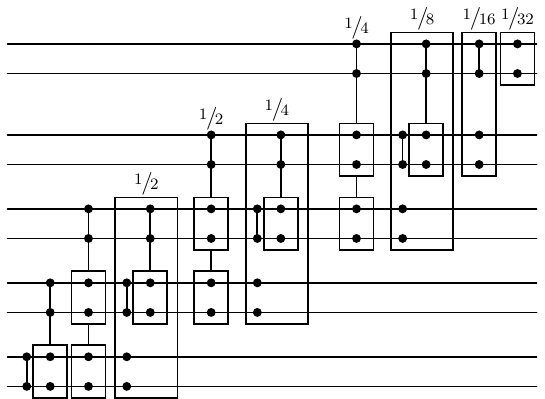}
  \\
  &\underset{(l=4)}{\overset{(\ref{eqn:merge_zigzag})}{=}}
  \circuitimg{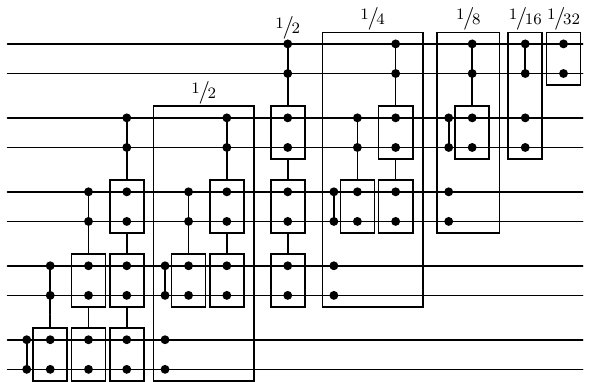}
  \\
  &\underset{(l=5)}{\overset{(\ref{eqn:merge_zigzag})}{=}}
  \circuitimg{tikz-cache/section_3__examples__merge_zigzag__3.pdf}.
\end{align*}